%% file: main.tex
\documentclass[letterpaper, 10 pt, conference]{ieeeconf}
\IEEEoverridecommandlockouts                              
\usepackage[margin=0.75in]{geometry}
\usepackage{graphicx}  
\input{preamble/preamble}

\title{\LARGE \bf  Heterogeneous Multi-Agent Multi-Target Tracking using Cellular Sheaves}
\author{Tyler Hanks$^{\ast}$, Cristian F. Nino$^{\ast}$, Joana Bou Barcelo, Austin Copeland, Warren Dixon, James Fairbanks
\thanks{$^\ast$These authors contributed equally.}
\thanks{Tyler Hanks is with the Department of Computer and Information Science and Engineering, and Cristian F. Nino, Joana Bou Barcelo, Austin Copeland, Warren Dixon, and James Fairbanks are with the Department of Mechanical and Aerospace Engineering, University of Florida, emails:
        {\tt\small \{t.hanks, cristian1928, jboubarcelo, acopeland1, wdixon, fairbanksj\}@ufl.edu}.} 
\thanks{Hanks was supported by the NSF Graduate Research Fellowship (DGE-1842473). Hanks, Bou Barcelo, and Fairbanks were supported by DARPA (HR00112220038); Hanks and Fairbanks also by ONR (N00014-23-1-2339). Nino, Copeland, and Dixon were supported by AFRL (FA8651-24-1-0018), AFOSR (FA9550-19-1-0169), and ARL (W911NF-25-2-0045). Any opinions, findings, and conclusions or recommendations are those of the authors and do not necessarily reflect the views of the NSF or other sponsors.}
}

\begin{document}
\maketitle
\thispagestyle{empty}
\pagestyle{empty}

\begin{abstract}
    Multi-agent target tracking in the presence of nonlinear dynamics and agent heterogeneity, where state-space dimensions may differ, is a challenging problem that traditional graph Laplacian methods cannot easily address. This work leverages the framework of cellular sheaves, a mathematical generalization of graph theory, to natively model such heterogeneous systems. While existing coordination sheaf frameworks focus on cooperative problems like consensus, this work extends them to the non-cooperative target-tracking problem. The tracking of multiple, unknown targets is formulated as a harmonic extension problem on a cellular sheaf, accommodating nonlinear dynamics and external disturbances for all agents. A decentralized control law is developed using the sheaf Laplacian, and a corresponding Lyapunov-based stability analysis is provided to guarantee tracking error convergence, with results validated by simulation.
\end{abstract}

\section{Introduction}

The study of multi-agent systems (MAS) has expanded significantly, where the central goal is to design local interaction rules that produce a desired collective behavior. Foundational research established the principles of distributed coordination through canonical problems such as consensus, where agents agree on a common value using only local neighbor information (e.g., \cite{Olfati-Saber.R2004}, \cite{Olfati-Saber.Reza2007}), as well as formation control (e.g., \cite{Mesbahi.Mehran2010}, \cite{Ren.Wei2013}). Many of these tasks are formulated within a leader-follower architecture, where followers must coordinate to track the leaders' states (e.g., \cite{Fax.J2004}, \cite{Hong.Yiguang2006}), providing a natural abstraction for numerous real-world coordination problems.

This paradigm extends directly to the multi-agent target-tracking problem, a comprehensive abstraction for surveillance and monitoring where the target is treated as a dynamic, non-cooperative leader (e.g., \cite{Nino.Cristian2025}). Developing effective solutions, however, requires moving beyond the ideal conditions, such as linear dynamics and complete information, assumed in much of the initial literature. Real-world systems have nonlinear dynamics and must be robust to unmodeled dynamics and external disturbances. Furthermore, information architecture is a critical constraint. Centralized control, which relies on a single processor with global information, is not scalable and has a single point of failure. These characteristics motivate a shift to distributed control, where agents rely on local inter-agent communication, or fully decentralized control, where agents must operate using only their own onboard sensor measurements.

The case-by-case design common in single-target tracking becomes difficult to scale when addressing more complex operational scenarios, such as multi-target assignments or heterogeneous objectives. A primary challenge in modern systems is the integration of heterogeneous agents, which may possess fundamentally distinct dynamical models, computational capabilities, or operational time scales. This heterogeneity can be extreme, with agents operating in different domains (e.g., air, ground, and maritime environments) and thus existing in different state-space dimensions. Developing robust and scalable control frameworks that can manage multiple objectives across such a diverse network remains a significant open problem.

These challenges have motivated the application of cellular sheaves and sheaf Laplacians as an organizing and unifying framework for multi-agent control. Cellular sheaves and sheaf Laplacians are a generalization of graphs and graph Laplacians, encompassing constructions such as scalar and matrix weighted graphs, signed graphs, and connection graphs (e.g, \cite{Hansen.Jakob2020a}). A key benefit of the cellular sheaf approach is in enabling heterogeneity in agent state spaces as well as inter-agent communication and coordination goals. This has been leveraged in diverse applications ranging from opinion dynamics (e.g., \cite{Hansen.Jakob2021}), gossip (e.g., \cite{Riess.Hans2022}), economics (e.g., \cite{Riess.Hans2023}), distributed optimization (e.g., \cite{Hansen.Jakob2019b}), and machine learning (e.g., \cite{Hansen.Jakob2020b, Bodnar.Cristian2022}). In the controls context, these ideas were applied to develop a general framework for the analysis and implementation of multi-agent coordination. The resulting coordination sheaves have the ability to incorporate heterogeneous agent state spaces as well as heterogeneous combinations of common coordination tasks such as consensus, flocking, and formation (e.g., \cite{Hanks.Tyler2025,Zhao.Yichen2025}).

\subsection*{Contributions}
The work on coordination sheaves thus far only considers the fully cooperative setting, where all agents work together to achieve a common goal. This work extends the coordination sheaf framework to encompass the problem of cooperatively tracking a set of unknown targets. This problem is formulated as a time-varying harmonic extension problem on a cellular sheaf which describes the inter-agent and agent-to-target communication and sensing models, and a decentralized controller is developed for solving this problem. The stability of the resulting closed-loop system is dependent on the \emph{relative cohomology} of the cellular sheaf on which the problem is defined. Cohomology is a powerful algebraic-topological invariant of cellular sheaves which is readily computable using linear algebra (e.g., \cite{Hansen.Jakob2019a}). Working within the cellular sheaf framework brings the same benefits of heterogeneity realized in the cooperative setting to the non-cooperative tracking setting. 

The core contributions are summarized as follows:
\begin{itemize}
    \item A formulation of the multi-agent multi-target tracking problem as a harmonic extension problem on a cellular sheaf,
    \item A generalization of this problem to networks of heterogeneous agents and targets, where state-space dimensions may differ,
    \item A decentralized controller for solving this problem, with stability conditioned on a relative cohomology calculation.
\end{itemize}

These contributions are validated with theoretical analysis and simulations.

\section{Preliminaries}
Let $N\in\mathbb{Z}_{>0}$ and $V\triangleq[N]=\{1,\ldots,N\}$. Consider a static, undirected, simple graph $G\triangleq(V,E)$, where $E\subseteq V\times V$ is irreflexive and symmetric. Write $i\leftrightarrow j \equiv ij$ to mean $(i,j)\in E$. For each $i\in V$, define the neighborhood $N_{i}\triangleq\{\,j\in V: i\leftrightarrow j\}$. Given a subset \(U\subseteq V\), the induced subgraph \(G|_{U}\triangleq(U,E|_{U})\) has edge set \(E|_{U}\triangleq\{(i,j)\in U\times U\mid i\leftrightarrow j\}\). Thus, for \(i,j\in U\) with \(i\neq j\), the adjacency relation in \(G|_{U}\) is identical to that in \(G\).

To describe heterogeneous local state spaces and interaction rules, the network is equipped with a cellular sheaf (e.g., \cite{Hansen.Jakob2020a,Hansen.Jakob2021}).

\begin{definition} \label{def:cellular-sheaf}
Given a graph $G = (V,E)$, a \define{(weighted) cellular sheaf} $\mathcal{F}$ assigns:
\begin{enumerate}
    \item[(a)] To each vertex $i\in V$, an inner product space $\mathcal{F}(i)$ called the \define{vertex stalk},
    \item[(b)] To each edge $ij\in E$, an inner product space $\mathcal{F}(ij)$ called the \define{edge stalk}, 
    \item[(c)] For each $i\in V$ and $j\in N_i$, a linear \define{restriction map} $\mathcal{F}_{i\face ij}\maps\mathcal{F}(i) \rightarrow \mathcal{F}(ij)$, where $i\face ij$ denotes that edge $ij$ is incident to vertex $i$.
\end{enumerate}
\end{definition}

Given a subgraph $G'=(V',E') \subseteq G$, the vertex and edge stalks over $G'$ assemble into spaces of \define{0-cochains} and \define{1-cochains} over $G'$ respectively:
\[
C^0(G';\FF)\triangleq \bigoplus_{i\in V'}\FF(i),\quad\quad C^1(G';\FF)\triangleq\bigoplus_{ij\in E'}\FF(ij),
\]
where $\bigoplus_{i\in I}H_i$ denotes the direct sum of inner product spaces $H_i$ over index set $I$.
The inner products on stalks extend via orthogonal direct sum to inner products on cochain spaces.

\begin{example}[Constant Sheaf]
    For a fixed $k\in\mathbb{R}_{>0}$ and $G=(V,E)$, the constant sheaf $\underline{\mathbb{R}}^k$ assigns $\R^k$ as every vertex and edge stalk, with identity restriction maps $\underline{\R}^k_{i\face ij} = I_{k}$ for all $i\in V, j\in N_i$, where $I_{k}$ denotes the $k\times k$ identity matrix.
\end{example}

The restriction maps of a cellular sheaf specify local linear constraints which propagate to a global system of linear equations over the whole graph on which the sheaf is defined. The local and global solutions of these systems are captured by notions of local and global sections.
\begin{definition}[Local and Global Sections]
    Suppose $\mathcal{F}$ is a cellular sheaf over $G$ with $G'=(V',E')$ a subgraph of $G$. 
    A \define{local section} of $\FF$ over $G'$ is a 0-cochain $\mathbf{x}\in C^0(G';\FF)$ such that $\FF_{i\face ij}x_i = \FF_{j\face ij}x_j$ for all $ij\in E'$.
    The space of local sections of $\FF$ over $G'$ is denoted $\Gamma(G';\FF)\subseteq C^0(G';\FF)$. A \define{global section} of $\FF$ corresponds to the case when $G'=G$.
\end{definition}

\begin{example}
    On the constant sheaf $\underline{\R}^k$ over a connected graph $G$, global sections correspond to consensus vectors, i.e., the same vector assigned to each vertex.
\end{example}

To compute global sections, tools from sheaf cohomology are used. Cohomology transforms questions of topology into questions of linear algebra and is thus readily computable.
\begin{definition}
    Given a cellular sheaf $\FF$ over a graph $G$, the \define{coboundary operator} is a linear map
$C^0(G; \mathcal{F}) \xrightarrow{\delta_{\mathcal{F}}} C^1(G; \mathcal{F})$ given by $(\delta_{\mathcal{F}} \mathbf{x})_{ij} = \mathcal{F}_{i \face ij}(x_i) - \mathcal{F}_{j \face ij}(x_j)$. The \define{sheaf Laplacian} of $\FF$ is a linear map $L_\FF\maps C^0(G;\FF)\to C^0(G;\FF)$ defined by $L_\FF\triangleq \delta_\FF^\top\delta_\FF$.
\end{definition}
By construction, the sheaf Laplacian is a symmetric positive semi-definite matrix. It has a block structure with diagonal blocks $L_{i,i}=\sum_{j\in N_i}\FF^\top_{i\face ij}\FF_{i\face ij}$ and off-diagonal blocks given by $L_{i,j} = -\FF^\top_{i\face ij}\FF_{j\face ij}$ whenever $ij\in E$ and $0$ otherwise. The sheaf Laplacian acts vertexwise as
\[
(L_\mathcal{F} \mathbf{x})_i = \underset{j\in N_i}{\sum} \mathcal{F}_{i \face ij}^\top  \bigl( \mathcal{F}_{i \face ij}(x_i) - \mathcal{F}_{j \face ij}(x_j) \bigr).
\]
The \define{degree-0} and \define{degree-1 cohomology} of a cellular sheaf is given by 
\[
H^0(G; \mathcal{F}) \triangleq \ker \delta_{\mathcal{F}}, \quad H^1(G; \mathcal{F}) \triangleq C^1(G; \mathcal{F})/ \image \delta_{\mathcal{F}}.
\]
It follows that $H^0(G; \mathcal{F})= \Gamma(G; \mathcal{F})$ because $x \in \ker \delta_{\mathcal{F}}$ when $\mathcal{F}_{i \face ij}(x_i) - \mathcal{F}_{j \face ij}(x_j) = 0$ for all $ij \in E$. This equivalency shows that the set of global sections has the structure of a vector subspace of $C^0(G;\FF)$. Moreover, $\ker \delta_\FF = \ker L_\FF$.

The notion of relative cohomology will also be used. Let $G'\subseteq G$ be a subgraph and define
\[
C^0(G,G';\FF)\triangleq \bigoplus_{i\in V(G)\setminus V(G')}\FF(i).
\]
The \define{degree-0 relative cohomology} of $G$ with respect to $G'$ is defined as
\[
H^0(G,G';\FF)\triangleq \ker \delta|_{C^0(G,G';\FF)}.
\]
The interpretation of an element $\mathbf{x}\in H^0(G,G';\FF)$ is as a global section of $\FF$ that vanishes on $G'$ \cite{Hansen.Jakob2021}.

\begin{example}
    Consider again the constant sheaf $\underline{\R}^k$ on a connected graph $G$. The sheaf Laplacian $L_{\underline{\R}^k}$ is equal to $L_G\otimes I_k$, where $L_G$ is the graph Laplacian of $G$. Given any subgraph $G'\subseteq G$, the relative cohomology of the constant sheaf is
    $H^0(G,G';\underline{\R}^k) = 0$
    because the only constant function on $G$ which vanishes on $G'$ is the 0 function.
\end{example}

\section{Problem Formulation}

\subsection{System Dynamics}
Let $N,M\in\mathbb{Z}_{>0}$. Consider a multi-agent system with controllable agents $V_q$ of size $|V_q|=N$ and target agents $V_p$ of size $|V_p|=M$. Denote the set of all $N+M$ agents as $V\triangleq V_q\sqcup V_p$, where $A\sqcup B$ denotes the disjoint union of sets $A$ and $B$.
For each controllable agent $i\in V_q$, the dynamics are
\begin{equation}
\dot{q}_{i}(t)=f_{i}\left(q_{i}(t),t\right)+g_{i}\left(q_{i}(t),t\right)u_{i}\left(t\right)+\omega_{i}\left(t\right),\label{eq:Agent Dynamics}
\end{equation}
where $t_0\in\mathbb{R}_{\geq0}$ denotes the initial time, $t\geq t_0$ denotes the current time, $q_i:\mathbb{R}_{\geq t_0}\to\mathbb{R}^{n_i}$ denotes the agents' unknown generalized position, $\dot{q}_i:\mathbb{R}_{\geq t_0}\to\mathbb{R}^{n_i}$ denotes the agents' unknown generalized velocity, the unknown functions
$f_{i}:\mathbb{R}^{n_i}\times\mathbb{R}_{\geq t_0}\to\mathbb{R}^{n_i}$ and
$\omega_{i}:\mathbb{R}_{\geq t_0}\to\mathbb{R}^{n_i}$ represent drift dynamics and exogenous disturbances, respectively, $g_{i}:\mathbb{R}^{n_i}\times\mathbb{R}_{\geq t_0}\to\mathbb{R}^{n_i\times s_{i}}$ denotes a known control effectiveness matrix, $u_{i}:\mathbb{R}_{\geq t_0}\to\mathbb{R}^{s_{i}}$ denotes the control input, and $s_{i}\in\mathbb{Z}_{>0}$ denotes the number of control channels.

For each target agent $k\in V_p$, the dynamics are
\begin{equation}
\dot{p}_k(t)=f_k\left(p_{k}(t),t\right)+\omega_{k}\left(t\right),\label{eq:Target Dynamics}
\end{equation}
where $p_k:\mathbb{R}_{\geq t_0}\to\mathbb{R}^{n_k}$ denotes the targets' unknown generalized position, $\dot{p}_k:\mathbb{R}_{\geq t_0}\to\mathbb{R}^{n_k}$ denotes the targets' unknown generalized velocity, and the unknown functions
$f_{k}:\mathbb{R}^{n_k}\times\mathbb{R}_{\geq t_0}\to\mathbb{R}^{n_k}$ and
$\omega_{k}:\mathbb{R}_{\geq t_0}\to\mathbb{R}^{n_k}$ represent modeling
uncertainties and exogenous disturbances, respectively. 
For notational convenience when ranging over controlled and target agents, set for all $i\in V$
\[
x_i\triangleq \begin{cases}
    q_i, & \text{if }i\in V_q,\\
    p_i, & \text{if }i\in V_p.
\end{cases}
\]

\begin{remark}
    Allowing agent-specific state dimensions ($x_i\in \mathbb{R}^{n_i}$) enables heterogeneous teams. For instance, a group of unmanned aerial vehicles (UAVs) evolving in $\mathbb{R}^3$ to track an unmanned surface vehicle (USV) in $\R^2$, and conversely. 
\end{remark}

Regularity and boundedness conditions used for well-posedness in the subsequent analysis are stated next.
\begin{assumption}
\label{functionbound}The functions $f_{i}$ are of class $\mathcal{C}^{1}$
for all $i\in V$.
\end{assumption}

\begin{assumption}
\label{control effectiveness}The functions $g_{i}$ are full row rank for all $(q_i,t)\in\mathbb{R}^{n_{i}}\times\mathbb{R}_{\geq{t_0}}$ and are of class $\mathcal{C}^{0}$. Furthermore, for each fixed $q_{i}$, the map $t\mapsto g_{i}(q_{i},t)$ is uniformly bounded in $t\in \mathbb{R}_{\geq{t_0}}$ for all $i\in V_q$.
\end{assumption}

\begin{assumption}
\label{disturbance bound}The functions $\omega_{i}$ are of class $\mathcal{C}^{0}$, and there exist known constants $\overline{\omega}_{i}\in\mathbb{R}_{>0}$ such that $\| \omega_{i}\left(t\right)\| \leq\overline{\omega}_{i}$ for all $t\in\mathbb{R}_{\geq t_{0}}$ and for all $i\in V$.
\end{assumption}

By \cite{Penrose.R1955}, Assumption \ref{control effectiveness} ensures that the pseudoinverses $g_{i}^{+}:\mathbb{R}^{n_i}\times\mathbb{R}_{\geq t_{0}}\to\mathbb{R}^{s_{i}\times n_i}$ act as right inverses for all $i\in V_q$ and $t\in\mathbb{R}_{\geq t_{0}}$. In addition, the following condition is imposed.
\begin{assumption}
\label{gInvBound}The functions $g_{i}^{+}$ are of class $\mathcal{C}^{0}$, and for each fixed $q_{i}$ the map $t\mapsto g_{i}^{+}(q_{i},t)$ is uniformly bounded in $t\in \mathbb{R}_{\geq{t_0}}$ for all $i\in V_q$.
\end{assumption}

\subsection{Communication Topology}

Let $G_q=(V_q,E_q)$ be a graph whose vertices are controllable agents and whose edges denote bidirectional communication. Similarly, let $G_p=(V_p, E_p)$ be the bidirectional communication graph among the target agents. Define the total graph $G=(V,E)$ by taking the disjoint union $G_q\sqcup G_p$ and then adding edges between $i\in V_q$ and $k\in V_p$ when agent $i$ can sense target $k$. The following assumption on $G$ ensures that each target is sensed by at least one agent and that the agent team can, collectively, track all targets.

\begin{assumption}[Communication topology]
\label{Communication Topology}The agent graph $G_q$ is connected, and for each target $k\in V_p$ there exists at least one node $i\in V_q$ with $ik\in E$.
\end{assumption}

A cellular sheaf specifies inter-agent communication models. Let $\FF$ be a cellular sheaf over $G$ with vertex stalks equal to agent state spaces, $\FF(i) = \R^{n_i}$ for $i\in V$.
Define the subsheaf of controllable agents by $\mathcal{Q}~\triangleq~\FF|_{G_q}$.

The sensing model between agents $i\in V$ and $j\in N_i$ is specified by ``parallel transport" on $\FF$. Explicitly, agents $i$ and $j$ can measure
\[\begin{aligned}
d_{i,j}&\triangleq\FF_{i\face ij}^\top(\FF_{i\face ij}x_i - \FF_{j\face ij}x_j),\\ d_{j,i}&\triangleq\FF_{j\face ij}^\top(\FF_{j\face ij}x_j - \FF_{i\face ij}x_i),
\end{aligned}
\]
which, in the constant-sheaf case, reduce to relative-position measurements. The next examples illustrate heterogeneous dimensions and local frames.

\begin{example}[USV-UAV projection]\label{ex:projections}
    Suppose $i\in V$ is a USV with state in $\R^2$ and agent $j\in N_i$ is a UAV with state in $\R^3$, with restriction maps
    \[
    \FF_{i\face ij} = \begin{bmatrix}
        1 & 0 \\ 0 & 1
    \end{bmatrix}\quad\quad\quad\FF_{j\face ij} = \begin{bmatrix}
        1 & 0 & 0\\ 0 & 1 & 0
    \end{bmatrix}.
    \]
    Thus, $i$ senses the difference between its planar position and the projection of $j$'s position onto the $(x,y)$-plane; conversely, $j$ senses the difference between its own position and $i$'s position embedded as $(x_i , y_i, 0)$.
\end{example}

For an edge $ij\in E$, the restriction maps $\FF_{i\face ij}\in\R^{m_{ij}\times n_i}$ and $\FF_{j\face ij}\in\R^{m_{ij}\times n_j}$ transport agent states into a common local measurement frame (stalk) $\FF(ij)=\R^{m_{ij}}$. The transported relative state
\[
y_{ij}\;\triangleq\;\FF_{j\face ij}x_j-\FF_{i\face ij}x_i\ \in\ \FF(ij)
\]
is measured via its range $r_{ij}\;\triangleq\;\|y_{ij}\|$ and bearing $b_{ij}\;\triangleq\;y_{ij}/\|y_{ij}\| \in\mathbb{S}^{m_{ij}-1}$, such that $y_{ij}=r_{ij}\,b_{ij}$. The corresponding node-wise pullbacks
\[
d_{i,j}=-\,\FF_{i\face ij}^\top y_{ij},
\qquad
d_{j,i}=\,\FF_{j\face ij}^\top y_{ij},
\]
realize decentralized residuals consistent with this range–bearing sensing.

\begin{example}[Aligned local frame]
Assume $n_i=n_j=n$ and $\FF(ij)=\R^n$. Let $R_i\in\mathrm{SO}(n)$ map global coordinates to agent $i$’s local frame. Choose the edge frame to coincide with $i$’s frame and set
\[
\FF_{i\face ij}=R_i,\qquad \FF_{j\face ij}=R_i.
\]
Then
\begin{small}
\[
y_{ij}=R_i(x_j-x_i),
\quad
r_{ij}=\|x_j-x_i\|,
\quad
b_{ij}=\frac{R_i(x_j-x_i)}{\|x_j-x_i\|},
\]
\end{small}so $r_{ij}$ is the Euclidean distance and $b_{ij}$ is the unit bearing from $i$ to $j$ expressed in $i$’s local coordinates.
\end{example}

\begin{example}[Planar sensing of 3D target]
Let $n_i=2$, $n_j=3$, and $\FF(ij)=\R^2$. Let $R_i\in\mathrm{SO}(2)$ map global coordinates to agent $i$’s local frame, and let $P\in\R^{2\times 3}$ project $\R^3$ onto the $(x,y)$-coordinates. Define
\[
\FF_{i\face ij}=R_i,\qquad \FF_{j\face ij}=R_i\,P.
\]
Then
\begin{small}
\[
y_{ij}=R_i\big(P\,x_j - x_i\big),\quad
r_{ij}=\|y_{ij}\|,\quad
b_{ij}=\frac{y_{ij}}{\|y_{ij}\|}\in\mathbb{S}^1,
\]
\end{small}so agent $i$ measures the planar range and bearing to the projection of $j$’s position onto the $(x,y)$-plane, expressed in $i$’s frame.
\end{example}

\begin{remark}
    (i) For the constant sheaf with identity restrictions, $y_{ij}=x_j-x_i$, recovering relative-position sensing. (ii) The choice of edge frame is absorbed into the restriction maps; any orthonormal reassignment is equivalent up to a left action on $\FF(ij)$. (iii) In decentralized implementations, each agent evaluates $\FF_{i\face ij}$ using only its local pose estimate; range–bearing sensors provide $(r_{ij},b_{ij})$ from which $y_{ij}$ is reconstructed without global coordinates.
\end{remark}

\subsection{Control Objective}

For notational convenience, define the ensemble controllable and target state vectors
\[\begin{aligned}
    \mathbf{q} \triangleq& \begin{bmatrix}
        q_1^\top & q_2^\top & \dots & q_N^\top
    \end{bmatrix}^\top \in C^0(G_q,\FF)\\
    \mathbf{p}\triangleq& \begin{bmatrix}
        p_1^\top & p_2^\top &\dots & p_M^\top
    \end{bmatrix}^\top\in C^0(G_p,\FF).
\end{aligned}
\]

The objective is to design a decentralized controller for each agent
$i\in V_q$ that tracks the targets using only the relative measurement model. Equivalently, require the residuals to vanish, $d_{i,j} = d_{j,i} = 0$ for all $ij\in E$, which holds precisely when
\[\mathbf{x}\triangleq\begin{bmatrix}
    \mathbf{q}^\top & \mathbf{p}^\top
\end{bmatrix}^\top\in C^0(G;\FF)\]
is a global section of $\FF$. Because the target states are uncontrolled, this constraint is generally infeasible. For example, an agent tasked with simultaneously tracking the positions of two targets not at consensus cannot achieve a 0 residual with respect to both targets. A tractable relaxation is obtained via harmonic extension.
\begin{definition}
    Let $\FF$ be a cellular sheaf on $G=(V,E)$, let $U\subseteq V$, and let $\mathbf{y}\in C^0(G|_U;\FF)$. Define $W \triangleq V\setminus U$. A \define{harmonic extension} of $\mathbf{y}$ to $G$ is a 0-cochain $\mathbf{x}\in C^0(G;\FF)$ satisfying
    \[
    (L_\FF\mathbf{x})|_W = 0,\qquad \mathbf{x}|_U = \mathbf{y}.
    \]
\end{definition}
Harmonic extensions exist for any local 0-cochain; when $H^0(G,G|_U;\FF)=0$, the harmonic extension of any local 0-cochain is unique (see~\cite[Proposition~4.1]{Hansen.Jakob2019a}). 

\begin{example}
    Consider three agents with states in $\R$ where agent 2 is tasked with tracking agents 1 and 3. This setup gives a constant sheaf $\underline{\R}$ on the 3-vertex path graph $P_3$. The sheaf Laplacian equals the graph Laplacian
    \[
    L_G = \begin{bmatrix}
        1 & -1 & 0\\
        -1 & 2 & -1 \\
        0 & -1 & 1
    \end{bmatrix}.
    \]
    Fix endpoint values $x_1 = a$ and $x_3 = c$. The harmonic extension is $\begin{bmatrix}
        a & b & c
    \end{bmatrix}^\top$
    where $b = \frac{a+c}{2}$. Interpreted as target tracking, if agent 2 tracks agents 1 and 3 that do not converge to consensus, the optimal steady state is the midpoint between them.
\end{example}

With the harmonic extension framework in place, the control objective can be stated precisely.

\begin{problem}\label{prob:1}
    Design a decentralized controller for each agent $i\in V_q$ such that the ensemble state $\mathbf{q}(t)$ tracks the harmonic extension of $\mathbf{p}(t)$ for a given trajectory of the target agents.
\end{problem}

\begin{remark}
    This objective is analogous to heat diffusion on a graph with Dirichlet boundary conditions. Consider the discrete heat flow $\dot{x}=-L_\FF x$ on interior vertices with boundary values fixed on $U\subseteq V$. The unique equilibrium satisfies $L_\FF x=0$ on $W=V\setminus U$ with $x|_U=y$, which is exactly the harmonic extension of $y$ to $G$. Equivalently, the harmonic extension is the unique minimizer of the Dirichlet energy $\frac{1}{2} \|\delta_\FF x\|^2$ subject to the boundary constraint $x|_U=y$.
    In the tracking context, the target states $\mathbf{p}(t)$ serve as time-varying Dirichlet data on $G_p$. The desired agent state $\mathbf{q}(t)$ should follow the harmonic extension of $\mathbf{p}(t)$ to the agent vertices, denoted $\mathbf{q}^*(t)$. When the boundary signal varies slowly relative to the closed-loop convergence rate, $\mathbf{q}(t)$ remains close to $\mathbf{q}^\ast(t)$, the quasi-static regime. More generally, the closed-loop behaves as a diffusion process driven by a time-varying boundary, producing a spatial interpolation of the target trajectories over the communication topology consistent with the sheaf constraints. 
\end{remark}

\section{Control Design}

Let $i\in V_q$ be arbitrary. For notational convenience, let $A_i\triangleq N_i\cap V_q$ denote the set of adjacent control agents to $i$ and $T_i\triangleq N_i\cap V_p$ denote the set of adjacent target agents to $i$. To quantify the control objective for agent $i$, define the disagreement $\eta_i\triangleq (L_\FF\mathbf{x})_i\in \R^{n_i}$. Expanding this definition yields
\begin{equation}\label{eq:eta_i}
\begin{aligned}
    \eta_i =& \sum_{j\in N_i}\FF_{i\face ij}^\top(F_{i\face ij}q_i - \FF_{j\face ij}x_j),\\
    =& \sum_{j\in A_i}\FF_{i\face ij}^\top(F_{i\face ij}q_i - \FF_{j\face ij}q_j) \\+& \sum_{k\in T_i}\FF^\top_{i\face ik}\FF_{i\face ik}q_i - \sum_{k\in T_i}\FF^\top_{i\face ik}\FF_{k\face ik}p_k.
\end{aligned}
\end{equation}

Based on the subsequent stability analysis, the control input is designed
as
\begin{equation}\label{eq: controller-1}
u_{i} =-k_{1}g_{i}^{+}\eta_{i},
\end{equation}
where $k_{1}\in\mathbb{R}_{>0}$ is a user-defined constant.


To facilitate analysis, we make the following definitions for the \emph{ensemble representation} of our system. 
Define the agent-target generalized degree map $D\maps C^0(G_q;\FF)\to C^0(G_q;\FF)$ as a block diagonal matrix whose $i$th block is $\sum_{k\in T_i}\FF^\top_{i\face ik}\FF_{i\face ik}q_i$.
Note that the $i$th block is 0 if agent $i$ cannot sense any targets.
Also, define the agent-target interaction map $B\maps C^0(G_p;\FF)\to C^0(G_q;\FF)$ as a block matrix with blocks $B_{i,k}\triangleq \FF_{i\face ik}^\top\FF_{k\face ik}$ if agent $i$ can sense target $k$ and 0 otherwise. Lastly, define the total agent interaction map $\mathcal{H}\maps C^0(G_q;\FF)\to C^0(G_q;\FF)$ as $\mathcal{H}\triangleq L_\mathcal{Q} + D$.
With these defined, it is possible to rewrite \eqref{eq:eta_i} as
\[
\eta_i = (L_\mathcal{Q}\mathbf{q})_i + (D\mathbf{q})_i - \sum_{k\in T_i}B_{i,k}p_k.
\]
Thus, stacking all the $\eta_i$'s yields the ensemble representation of disagreement across the whole network
\begin{equation}\label{eta1}
\eta = (L_\mathcal{Q} + D)\mathbf{q} - B\mathbf{p} = \mathcal{H}\mathbf{q} - B\mathbf{p}.
\end{equation}
The next lemma confirms that regulating $\eta$ to 0 is equivalent to solving for a harmonic extension of the target state $\mathbf{p}$.
\begin{lemma}\label{lem:eta0_harmonic}
    If $\eta=0$, then $\begin{bmatrix}
        \mathbf{q}^\top & \mathbf{p}^\top
    \end{bmatrix}^\top$ is a harmonic extension of $\mathbf{p}$ to $G$.
\end{lemma}
\begin{proof}
    Following the discussion in \cite[Section 3.1]{Hansen.Jakob2020a}, the problem of harmonic extension can be written using a block decomposition of $L_\FF$ as
    \[
    \begin{bmatrix}
    L_\FF[V_q,V_q] & L_\FF[V_q, V_p] \\ L_\FF[V_p, V_q] & L_\FF[V_p, V_p]
    \end{bmatrix}
    \begin{bmatrix}
        \mathbf{q}\\\mathbf{p}
    \end{bmatrix} = \begin{bmatrix}
        0\\\mathbf{z}
    \end{bmatrix},
    \]
    where $\mathbf{z}$ is a free variable. Therefore, if
    \[
    L_\FF[V_q,V_q]\mathbf{q} + L_\FF[V_q, V_p]\mathbf{p} = 0,
    \]
    then $\begin{bmatrix}
        \mathbf{q}^\top & \mathbf{p}^\top
    \end{bmatrix}^\top$ is a harmonic extension of $\mathbf{p}$ to $G$. Observing that $L_\FF[V_q,V_q] = \mathcal{H}$ and $L_\mathcal{F}[V_q, V_p]= -B$ completes the proof.
\end{proof}

Lemma \ref{lem:eta0_harmonic} implies that designing a decentralized controller to regulate $\eta$ to 0 solves Problem \ref{prob:1}. This holds when 
\begin{equation}\label{eq:harmonic_equation}
    \mathcal{H}\mathbf{q} = B\mathbf{p}.
\end{equation}
To ensure uniqueness of solutions to \eqref{eq:harmonic_equation}, we require the following assumption.

\begin{assumption}\label{as:relative_cohom}
    The relative cohomology of $G$ with respect to $G_p$ is trivial. Formally, $H^0(G, G_p;\FF) = 0$.
\end{assumption}

Assumption \ref{as:relative_cohom} yields the following result about~\eqref{eq:harmonic_equation}.
\begin{lemma}
    The matrix $\mathcal{H}$ is invertible and positive definite.
\end{lemma}
\begin{proof}
    $\mathcal{H}$ is positive semi-definite because it is the sum of positive semi-definite matrices. Assumption \ref{as:relative_cohom} ensures that solutions of \eqref{eq:harmonic_equation} exist uniquely for any $\mathbf{p}\in C^0(G_p,\FF)$, implying that $\mathcal{H}$ is invertible. Because $\mathcal{H}$ is positive semi-definite and invertible, it is positive definite.
\end{proof}

\begin{remark}
    Assumption \ref{as:relative_cohom} is crucial, as it eliminates ambiguity in the tracking task and, as will be shown, ensures stability of the closed-loop control system. The trade-off for this guarantee is a restriction on the types of heterogeneity permitted within the framework. For illustration, consider the setup of Example \ref{ex:projections}. With these restriction maps, surface agents could track the projection of aerial targets onto the $xy$-plane. This is because, for any target positions, there is a unique harmonic extension (specifically, the centroid of the projected target states). Conversely, aerial agents could not track ground-based targets because the harmonic extension of the grounded target states to the aerial agent states does not impose any constraints on the agent's height ($z$-coordinate). In this failing scenario, the cohomology of the whole graph relative to the USV subgraph is $\R$ rather than $0$. This non-zero cohomology arises from the one-dimensional agreement subspace (the unconstrained $z$-axis), which violates Assumption~\ref{as:relative_cohom}.
\end{remark}

Based on the subsequent analysis, define the harmonic extension $\mathbf{q}^* \in C^0(G_q,\FF)$ for a given target configuration $\mathbf{p}\in C^0(G_p,\FF)$ as
\begin{equation}\label{qStar}
   \mathbf{q}^* \triangleq \mathcal{H}^{-1}B\mathbf{p}, 
\end{equation}
and the tracking error $\mathbf{e}\in C^0(G_q,\FF)$ as
\begin{equation}\label{trackingError}
  \mathbf{e} \triangleq \mathbf{q}^* - \mathbf{q}.  
\end{equation}
Then, substituting \eqref{qStar} and \eqref{trackingError} into \eqref{eta1} yields
\begin{equation}\small{}
\label{eq:eta_he}
    \eta = \mathcal{H}(\mathbf{q}^*-\mathbf{e}) - B\mathbf{p}  = \mathcal{H}(\mathcal{H}^{-1}B\mathbf{p}-\mathbf{e}) - B\mathbf{p} = -\mathcal{H}\mathbf{e}.
\end{equation}

Now, define the ensemble agent dynamics $f_q\maps C^0(G_q,\FF)\times\R_{\geq 0}\to C^0(G_q,\FF)$ and disturbances $\omega_q\maps \R_{\geq 0}\to C^0(G_q,\FF)$ as 
\[
f_q(\mathbf{q},t)\triangleq \bigl[f_i(q_i,t)\bigr]_{i\in V_q},\quad \omega_q(t)\triangleq \bigl[\omega_i(t)\bigr]_{i\in V_q}.
\]
Similarly, define the ensemble target dynamics $f_p\maps C^0(G_p,\FF)\times\R_{\geq 0}\to C^0(G_p,\FF)$ and disturbances $\omega_p\maps \R_{\geq 0}\to C^0(G_p,\FF)$ as
\[
f_p(\mathbf{p},t)\triangleq \bigl[ f_i(p_i,t) \bigr]_{i\in V_p},\quad \omega_p(t)\triangleq \bigl[ \omega_i(t)\bigr]_{i\in V_p}.
\]

To ensure that the controller can achieve the goal, it is necessary to assume that the difference between the agent dynamics and target dynamics are bounded. This is addressed in the following assumption.

\begin{assumption}[Bounded Difference]\label{Bounded Difference}
    The mapping 
    \[\mathbf{p}\mapsto \mathcal{H}^{-1}Bf_p(\mathbf{p},t) - f_q(\mathcal{H}^{-1}B\mathbf{p},t)\]
    is of class $\mathcal{L}_\infty\bigl(C^0(G_p,\FF)\times\R_{\geq 0},C^0(G_q,\FF)\bigr)$.
\end{assumption}

\begin{figure}
    \centering
    \includegraphics[width=\linewidth]{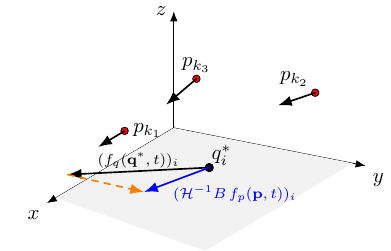}
    \caption{This figure illustrates the bounded difference between the agent and target dynamics. A surface agent $i$ is tasked with tracking three aerial targets $p_{k_j}$ for $j\in \{1,2,3\}$. The point $q_i^*$ denotes the harmonic extension of the target states to the surface domain; the blue tangent vector represents the harmonic extension of the targets’ tangent vectors. The orange vector indicates the discrepancy between the harmonically extended target dynamics and the agent dynamics evaluated at $q_i^*$. Assumption \ref{Bounded Difference} requires this discrepancy to be uniformly bounded for all target configurations.}
    \label{fig:assumption_fig}
\end{figure}

\begin{remark}
    Assumption \ref{Bounded Difference} quantifies the extent to which the following diagram
\[\begin{tikzcd}
	{C^0(G_p,\mathcal{F})} && {C^0(G_q,\mathcal{F})} \\
	{C^0(G_p,\mathcal{F})} && {C^0(G_q,\mathcal{F})}
	\arrow["{\mathcal{H}^{-1}B}", from=1-1, to=1-3]
	\arrow["{f_p}"', from=1-1, to=2-1]
	\arrow["{f_q}", from=1-3, to=2-3]
	\arrow["{\mathcal{H}^{-1}B}"', from=2-1, to=2-3]
\end{tikzcd}\]
    fails to commute. The paths in this diagram represent distinct ways of transporting target state vectors $\mathbf{p}$ to agent tangent vectors $\dot{\mathbf{q}}$ via harmonic extension. Explicitly, it asks that the harmonic extension of the target dynamics at $\mathbf{p}$ have bounded difference with the agent dynamics at the harmonic extension of $\mathbf{p}$. This is illustrated by Figure \ref{fig:assumption_fig}.
\end{remark}

Let $\mathbf{u}\triangleq [u_i]_{i\in V_q} \in \bigoplus_{i\in V_q}\R^{s_i}$. Then, using \eqref{eq: controller-1} and \eqref{eq:eta_he}, it follows that
\begin{equation}\label{ensembleController}
    \mathbf{u}=k_1 \mathbf{g}^+\mathcal{H}\mathbf{e},
\end{equation}
where $\mathbf{g}^+\maps C^0(G_q,\FF)\times\R_{\geq t_0}\to \bigoplus_{i\in V_q}\R^{s_i\times n_i}$ is defined as $\mathbf{g}^+\triangleq\text{blkdiag}\{g_1^+,\ldots,g_N^+ \}$.
Taking the time-derivative of \eqref{trackingError} and using \eqref{ensembleController} yields the closed-loop error system
\begin{equation}\label{eq:error_dynamics}
    \dot{\mathbf{e}} = -k_1\mathcal{H}\mathbf{e} + f_B + \tilde{f} + \Delta,
\end{equation}
where $f_B\triangleq \mathcal{H}^{-1}Bf_p(\mathbf{p},t) - f_q(\mathcal{H}^{-1}B\mathbf{p},t)$, $\tilde{f}\triangleq f_q(\mathbf{q}^*,t) - f_q(\mathbf{q},t)$, and $\Delta\triangleq \mathcal{H}^{-1}B\omega_p(t) - \omega_q(t)$.

\section{Stability Analysis}

The evolution of $\mathbf{e}$ is governed by the initial value problem
\begin{equation}
\dot{\mathbf{e}}=h\left(\mathbf{e},t\right),\ \mathbf{e}\left(t_{0}\right)=\mathbf{e}_{0},\label{eq:IVP}
\end{equation}
where $\mathbf{e}_{0}\in C^0(G_q,\FF)$
is the initial state and the vector field $h\maps C^0(G_q,\FF)\times\R_{\geq t_0}\to C^0(G_q,\FF)$ is defined as
\begin{equation}\label{eq:CLES}
h(\mathbf{e},t)\triangleq -k_1\mathcal{H}\mathbf{e} + f_B + \tilde{f} + \Delta.
\end{equation}

Following Assumptions \ref{functionbound} and \ref{disturbance bound}, $h$ is locally Lipschitz in $\mathbf{e}$ and continuous in $t$. Therefore, by the Picard-Lindelöf theorem (e.g., \cite[Chapter 1, Theorem 3.1]{Coddington.Earl1956}), for every $\mathbf{e}_{0}\in C^0(G_q,\FF)$ there exists a unique maximal solution. Consider the Lyapunov function candidate $V\maps C^0(G_q,\FF)\to\mathbb{R}_{\geq t_0}$
defined as
\begin{equation}
V\left(\mathbf{e}\right)\triangleq\frac{1}{2}\mathbf{e}^{\top}\mathbf{e}.\label{eq:Lyapunov}
\end{equation}

Using the fundamental theorem of calculus and Assumption \ref{functionbound},
there exists a function $\rho\maps\mathbb{R}_{\geq t_0}\to\mathbb{R}_{\geq t_0}$
such that $\| \tilde{f} \| \leq\rho (\|\mathbf{e}\|) \|\mathbf{e}\| $, where $\overline{\rho}\left(\cdot\right)\triangleq\rho\left(\cdot\right)-\rho\left(0\right) \in \mathcal{K}_\infty$. Furthermore, by Assumption
\ref{Bounded Difference}, there exists $\overline{f}\in\mathbb{R}_{\geq0}$
such that $\| f_{B} \| \leq\overline{f}$. Let $P = \max\{N,M\}$. Based on the subsequent analysis, define $\mu\triangleq\frac{\left(\overline{f}+\overline{\omega}\right)^{2}}{2k_{1}\lambda_{\min}\left\{ \mathcal{H}\right\} }$
and $k_{\min}\triangleq\frac{1}{2}k_{1}\lambda_{\min}\left\{ \mathcal{H}\right\} $, where $\overline{\omega}\triangleq \Bigl(1 + \frac{\sigma_\mathrm{max}\{B\}}{\lambda_\mathrm{min}\{\mathcal{H}\}}\Bigr)\sqrt{P}\max_{i\in V} \{\overline{\omega}_i \}$, which exists for all $i\in V$ and all $t\in\mathbb{R}_{\geq t_0}$
by Assumption \ref{disturbance bound}.

For notational convenience, let $\chi \triangleq \overline{\rho}^{-1}\left(k_{\min}-\lambda_{V}-\rho\left(0\right)\right)$, where $\lambda_{V}\in\mathbb{R}_{>0}$ is a user-defined parameter
which controls the desired rate of convergence. The region to which the
state trajectory is constrained is defined as
\begin{equation}
\mathcal{D}\triangleq\left\{ \iota\in C^0(G_q,\FF) : \left\Vert \iota\right\Vert \leq \chi \right\}.\label{eq:state space}
\end{equation}

For the dynamical
system described by (\ref{eq:IVP}), the set of initial conditions
is defined as
\begin{equation}
\begin{aligned}\mathcal{S} & \triangleq\left\{ \iota\in C^0(G_q,\FF) :\left\Vert \iota\right\Vert \leq \chi -\sqrt{\frac{\mu}{\lambda_{V}}}\right\} ,\end{aligned}
\label{eq:initial conditions}
\end{equation}
and the uniform ultimate bound is defined as
\begin{equation}
\mathcal{U}\triangleq\left\{ \iota\in C^0(G_q,\FF) :\left\Vert \iota\right\Vert \leq\sqrt{\frac{\mu}{\lambda_{V}}}\right\} .\label{eq:equilibrium set}
\end{equation}

\begin{theorem}
\label{thm:}Consider the dynamical system described by (\ref{eq:Agent Dynamics})
and (\ref{eq:Target Dynamics}). For any $\mathbf{e}_0 \in\mathcal{S}$,
the decentralized controller (\ref{eq: controller-1}) ensures that
$\mathbf{e}$ semi-globally uniformly exponentially converges to
$\mathcal{U}$ with estimate{\small{}
\[
\left\Vert \mathbf{e}\left(t\right)\right\Vert \leq\sqrt{{\rm e}^{-2\lambda_{V}\left(t-t_{0}\right)}\left\Vert \mathbf{e}_{0}\right\Vert ^{2}+\frac{\mu}{\lambda_{V}}\left(1-{\rm e}^{-2\lambda_{V}\left(t-t_{0}\right)}\right)},
\]
}for all $t\geq t_{0}$, provided that the sufficient gain condition
$k_{\min}>\rho\left(2\sqrt{\frac{\mu}{\lambda_{V}}}\right)+\lambda_{V}$
is satisfied and Assumptions \ref{functionbound}-\ref{Bounded Difference}
hold. Moreover, $\mathcal{U}\subset\mathcal{S}\subset\mathcal{D}$
and the solution $\mathbf{e}$ is complete.
\end{theorem}

\begin{proof}
Taking the derivative of (\ref{eq:Lyapunov}) along the trajectories
of (\ref{eq:IVP}) and using (\ref{eq:CLES}) yields
\begin{equation}
\begin{aligned}
\frac{{\rm d}}{{\rm d}t}V\left(\mathbf{e}\left(t\right)\right) 
 & =-k_{1}\mathbf{e}^{\top}\left(t\right)\mathcal{H}\mathbf{e}\left(t\right)+\mathbf{e}^{\top}\left(t\right)f_{B} \\
 & +\mathbf{e}^{\top}\left(t\right)\tilde{f}+\mathbf{e}^{\top}\left(t\right)\Delta.
\end{aligned}
\label{eq:vDot}
\end{equation}
Using the definitions of $\Delta$ and $\overline{\omega}$, the
triangle inequality yields
\begin{equation}\label{eq:DeltaBounds}
    \|\Delta\| \leq \Bigl(1 + \frac{\sigma_\mathrm{max}\{B\}}{\lambda_\mathrm{min}\{\mathcal{H}\}}\Bigr)\sqrt{P}\max_{i\in V} \{\overline{\omega}_i \} = \overline{\omega}.
\end{equation}

Applying the triangle and Cauchy-Schwarz inequalities and using (\ref{eq:DeltaBounds})
and the definitions of $\rho$ and $\overline{f}$ yields that (\ref{eq:vDot})
is bounded above as
\begin{equation}
\begin{aligned}
\frac{{\rm d}}{{\rm d}t}V\left(\mathbf{e}\left(t\right)\right) & \leq-k_{1}\lambda_{\min}\left\{ \mathcal{H}\right\} \left\Vert \mathbf{e}\left(t\right)\right\Vert ^{2} \\
 & +\rho\left(\left\Vert \mathbf{e}\left(t\right)\right\Vert \right)\left\Vert \mathbf{e}\left(t\right)\right\Vert ^{2} \\
 & +\left(\overline{f}+\overline{\omega}\right)\left\Vert \mathbf{e}\left(t\right)\right\Vert.
\end{aligned}
\label{eq:vDot-1}
\end{equation}

Completing the square yields that (\ref{eq:vDot-1}) is bounded above as
\begin{equation}\small{}
\frac{{\rm d}}{{\rm d}t}V\left(\mathbf{e}\left(t\right)\right) \leq -\left(k_{\min}-\rho\left(\left\Vert \mathbf{e}\left(t\right)\right\Vert \right)\right)\left\Vert \mathbf{e}\left(t\right)\right\Vert ^{2}+\mu.
\label{eq:vDot-2}
\end{equation}

It is now shown, by contradiction, that the maximal solution exists
for all time and remains in the interior of $\mathcal{D}$. Fix any $\mathbf{e}_{0}\in\mathcal{S}$
and let $\mathbf{e}:\left[t_{0},T_{\max}\right)\to C^0(G_q,\FF)$ be the
corresponding unique maximal Carathéodory solution of (\ref{eq:IVP}).

Let $\mathcal{I}\triangleq\left\{ t\in\left[t_{0},T_{\max}\right):\mathbf{e}\left(\tau\right)\in\mathcal{D}\text{ for all }\tau\in\left[t_{0},t\right]\right\} $, where $T_{\max}\in(t_{0},\infty ]$ denotes the supremum of all times $T>t_{0}$ for which a solution exists on $[t_{0},T)$. Since $\mathbf{e}_{0} \in\mathcal{S}\subset\mathcal{D}$
and $t\mapsto \mathbf{e}\left(t\right)$ is continuous, the set $\mathcal{I}$
is non-empty and contains an interval $\left[t_{0},t_{0}+\epsilon\right)$
for some $\epsilon>0$.

For all $t\in\mathcal{I}$, continuity implies there exists a time interval such that $\mathbf{e}\left(t\right)\in\mathcal{D}$. Using $\overline{\rho}\left(\cdot\right)=\rho\left(\cdot\right)-\rho\left(0\right)$
indicates that $\rho\left(\cdot\right)=\overline{\rho}\left(\cdot\right)+\rho\left(0\right)$. Then $\left\Vert \mathbf{e}\left(t\right)\right\Vert \leq\overline{\rho}^{-1}\left(k_{\min}-\lambda_{V}-\rho\left(0\right)\right)$
which implies $\rho\left(\left\Vert \mathbf{e}\left(t\right)\right\Vert \right)\leq k_{\min}-\lambda_{V}-\rho\left(0\right)$, 
i.e., $\rho\left(\left\Vert \mathbf{e}\left(t\right)\right\Vert \right)+\rho\left(0\right)-k_{\min}\leq-\lambda_{V}$.
Using this inequality and (\ref{eq:Lyapunov}) yields that  (\ref{eq:vDot-2}) is bounded above as
\begin{equation}
\frac{{\rm d}}{{\rm d}t}V\left(\mathbf{e}\left(t\right)\right) \leq-2\lambda_{V}V\left(\mathbf{e}\left(t\right)\right)+\mu,
\label{eq:vDot-5}
\end{equation}
for all $t\in\mathcal{I}$. Applying the Grönwall-Bellman inequality
to (\ref{eq:vDot-5}) yields
\begin{equation}
\begin{aligned}V\left(\mathbf{e}\left(t\right)\right) & \leq{\rm e}^{-2\lambda_{V}\left(t-t_{0}\right)}V\left(\mathbf{e}\left(t_{0}\right)\right) \\
&+\frac{\mu}{2\lambda_{V}}\left(1-{\rm e}^{-2\lambda_{V}\left(t-t_{0}\right)}\right),\end{aligned}
\label{eq:vsol}
\end{equation}
for all $t\in\mathcal{I}$. Applying (\ref{eq:Lyapunov}) to (\ref{eq:vsol}) yields
\begin{equation}\label{eq: final bound}
    \begin{aligned}
        \left\Vert \mathbf{e}\left(t\right)\right\Vert &\leq\Bigl({\rm e}^{-2\lambda_{V}\left(t-t_{0}\right)}\left\Vert \mathbf{e}_{0}\right\Vert ^{2}\\
        &+\text{ }\frac{\mu}{\lambda_{V}}\left(1-{\rm e}^{-2\lambda_{V}\left(t-t_{0}\right)}\right)\Bigr)^{\frac{1}{2}},
    \end{aligned}
\end{equation}
for all $t\in\mathcal{I}$.
Since $e\left(t_{0}\right)\in\mathcal{S}$, $\left\Vert \mathbf{e}\left(t_{0}\right)\right\Vert \leq\overline{\rho}^{-1}\left(k_{\min}-\lambda_{V}-\rho\left(0\right)\right)-\sqrt{\frac{\mu}{\lambda_{V}}}$.
Substituting this expression into (\ref{eq: final bound}) and using the
facts that ${\rm e}^{-2\lambda_{V}\left(t-t_{0}\right)}\leq1$ and
$1-{\rm e}^{-2\lambda_{V}\left(t-t_{0}\right)}<1$ for all $t\geq t_{0}$
yield the strict bound $\| \mathbf{e}\left(t\right) \| <\overline{\rho}^{-1}\left(k_{\min}-\lambda_{V}-\rho\left(0\right)\right)$, for all $t\in\mathcal{I}$. Thus $\mathbf{e}\left(t\right)\in\text{int}\left(\mathcal{D}\right)$
for all $t\in\mathcal{I}$. Now, assume for contradiction that the
maximal time of existence is finite, i.e., $\sup\mathcal{I}<T_{\max}$.
If $\sup\mathcal{I}<T_{\max}$, continuity of $\mathbf{e}\left(\cdot\right)$
would imply $\mathbf{e}\left(\sup\mathcal{I}\right)\in\partial\mathcal{D}$,
contradicting $\mathbf{e}\left(t\right)\in\text{int}\left(\mathcal{D}\right)$
for all $t\in\mathcal{I}$. Hence $\sup\left(\mathcal{I}\right)=T_{\text{max}}$
and $\mathcal{I}=\left[t_{0},T_{\max}\right)$, i.e., $\mathbf{e}\left(t\right)\in\mathcal{D}$
for all $t\in\left[t_{0},T_{\max}\right)$.

Since $\mathbf{e}\left(t\right)\in\mathcal{D}$ for all $t\in\left[t_{0},T_{\max}\right)$
and $\mathcal{D}$ is compact, $\sup_{t\in\left[t_{0},T_{\max}\right)}\left\Vert \mathbf{e}\left(t\right)\right\Vert <\infty$.
By \cite[Chapter 1, Theorem 4.1]{Coddington.Earl1956}, the
solution extends from $\left[t_{0},T_{\max}\right)$ to $\left[t_{0},\infty\right)$.
Therefore the solution exists for all $t\geq t_{0}$ with $\mathbf{e}\left(t\right)\in\mathcal{D}$
for all $t\geq t_{0}$. Consequently, for all $\mathbf{e}_{0} \in\mathcal{S}$,
(\ref{eq: final bound}) holds for every $t\geq t_{0}$.

As $t\to\infty$, the bound in (\ref{eq: final bound}) converges 
to $\left\Vert \mathbf{e}\left(t\right)\right\Vert \leq\sqrt{\frac{\mu}{\lambda_{V}}}$,
i.e., the solution is uniformly ultimately bounded with ultimate set
$\mathcal{U}$ by (\ref{eq:equilibrium set}). Furthermore, since
$\mu=\frac{\left(\overline{f}+\overline{\omega}\right)^{2}}{2k_{1}\lambda_{\min}\left\{ \mathcal{H}\right\} }$, $\mu$ can be arbitrarily decreased
by increasing $k_{1}$.

Next, recall that $k_{\min}>\rho\left(2\sqrt{\frac{\mu}{\lambda_{V}}}\right)+\lambda_{V}$.
Using the definition of $\overline{\rho}$ and the fact that $\overline{\rho}$
is invertible yields $\sqrt{\frac{\mu}{\lambda_{V}}}<\overline{\rho}^{-1}\left(k_{\min}-\lambda_{V}-\rho\left(0\right)\right)-\sqrt{\frac{\mu}{\lambda_{V}}}$,
i.e., $\mathcal{U}\subset\mathcal{S}$. Furthermore, since $\sqrt{\frac{\mu}{\lambda_{V}}}>0$,
strict monotonicity of $\rho$ gives $k_{\min}>\rho\left(2\sqrt{\frac{\mu}{\lambda_{V}}}\right)+\lambda_{V}>\rho\left(\sqrt{\frac{\mu}{\lambda_{V}}}\right)+\lambda_{V}$
which implies $k_{\min}-\lambda_{V}-\rho\left(0\right)>\overline{\rho}\left(\sqrt{\frac{\mu}{\lambda_{V}}}\right)>0$.
Hence, $\overline{\rho}^{-1}\left(k_{\min}-\lambda_{V}-\rho\left(0\right)\right)-\sqrt{\frac{\mu}{\lambda_{V}}}>0$,
so $\mathcal{S}$ has strictly positive radius i.e., $\mathcal{S}\neq\varnothing$.
Thus, $\mathcal{U}\subset\mathcal{S}\subset\mathcal{D}$. Moreover,
since $\overline{\rho}\in\mathcal{K}_{\infty}$, $\overline{\rho}^{-1}\in\mathcal{K}_{\infty}$.
Thus, since $k_{\min}=\frac{1}{2}k_{1}\lambda_{\min}\left\{ \mathcal{H}\right\} $,
letting $k_{1}\to\infty$ yields $\overline{\rho}^{-1}\left(k_{\min}-\lambda_{V}-\rho\left(0\right)\right)\to\infty$,
hence $\mathcal{D}$ can be made arbitrarily large and the result
is semi-global.

Since $\lambda_{V}$ is independent of the initial time $t_{0}$ or
initial condition $\mathbf{e}_{0}$, the exponential convergence
is uniform. Additionally, since $\mathbf{e}\in\mathcal{L}_\infty \bigl(\R_{\geq t_0};C^0(G_q,\FF)\bigr)$, using \eqref{eq:eta_he} yields $\eta\in\mathcal{L}_\infty \bigl(\R_{\geq t_0};C^0(G_q,\FF)\bigr)$. Thus, by (\ref{eq: controller-1})
and Assumption \ref{gInvBound}, $ \mathbf{u}\in \mathcal{L}_\infty \bigl(\R_{\geq t_0};C^0(G_q,\FF)\bigr) $ and all implemented signals are bounded.
\end{proof}

\section{Simulation}

A 13-agent planar network tracks four targets in \(\mathbb{R}^3\) that self-organize into a regular tetrahedron via decentralized consensus. Four agents form a square about the apex target; the remaining nine form three equilateral triangles, each centered on a base target. The background velocity field superposes point sources/sinks, axis-aligned vortices, and Gaussian axial flows. Point singularities are \(SS=\big[\begin{smallmatrix}-2.0&1.0&0.8&1.2\\[1pt]1.6&-1.2&-0.5&-1.0\\[1pt]0.2&2.0&1.6&0.6\end{smallmatrix}\big]\), where each row encodes \((p_x,p_y,p_z,\alpha)\) with \(\alpha>0\) a source and \(\alpha<0\) a sink. Vortical components are \(VV=\big[\begin{smallmatrix}-1.2&-0.6&0.0&0.0&0.0&1.0&0.9\\[1pt]2.0&1.3&0.7&0.0&1.0&0.0&-0.7\end{smallmatrix}\big]\), where each row encodes \((q_x,q_y,q_z,a_x,a_y,a_z,\gamma)\); the sign of \(\gamma\) sets rotation about axis \(a\) via the right-hand rule. Gaussian axial components use centers \(c_k\), unit directions \(\hat d_k\), strengths \(\beta_k\), lateral widths \(w_k\), and axial lengths \(l_k\); at evaluation point \(X=(x,y,z)\), define \(r_k=X-c_k\), \(t_k=r_k^\top\hat d_k\), \(n_k=r_k-t_k\hat d_k\), and \(v_k=\beta_k\exp(-\|n_k\|^2/2w_k^2)\exp(-t_k^2/2l_k^2)\hat d_k\). Two such components are used: \(c=(-3.0,0.0,0.0)\), \(d=[1.0,0.0,0.1]\), \(\beta=0.9\), \(w=0.9\), \(l=5\); and \(c=(0.0,-2.0,1.0)\), \(d=[0.7,0.7,0.0]\), \(\beta=0.6\), \(w=0.55\), \(l=3.2\). Agents evolve in the \(xy\)-plane under the same field.

Target motion is a Lissajous-type reference \(\dot x=Aa\cos(at+\delta_x)\), \(\dot y=Bb\cos(bt)\), \(\dot z=Cc\cos(ct+\delta_z)\) with \(A=20\,\mathrm{m}\), \(B=10\,\mathrm{m}\), \(C=15\,\mathrm{m}\), \((a,b,c)=(1,2,1)\,\mathrm{rad\,s^{-1}}\), \(\delta_x=\pi/2\), \(\delta_z=0\). Agents start from arbitrary planar positions; targets initialize on the negative \(z\)-axis. Formations are enforced by constant offsets relative to assigned targets and inter-agent distance constraints. Squares place the target at the center; equilateral triangles place the target at the centroid with \(R=l/\sqrt{3}\) and \(60^\circ\) inter-edge angle. The tetrahedron uses base vertices \((0,0,0)\), \((d,0,0)\), \((\tfrac{d}{2},\tfrac{\sqrt{3}d}{2},0)\) and apex \((\tfrac{d}{2},\tfrac{\sqrt{3}d}{6},\tfrac{\sqrt{6}}{3}d)\). Edge-wise equal-and-opposite offset updates preserve the centroid, implying \(\sum_{i=1}^{4}\mathrm{offset}_i=0\).

With proportional gains \(k_{p,A}=10\) (agents) and \(k_{p,T}=5\) (targets), agent groups converge to the commanded square/triangular formations about their assigned targets, and the targets converge to a regular tetrahedron. Sparse cross-links between formations yield small steady offsets from exact centering. Planar tracking produces figure-8 agent trajectories in \(xy\), consistent with the reference. Fig. \ref{fig: subplots} is a comparison of the initial vs final configuration of the simulation. As seen, all agents successfully achieve their corresponding formation about their desired target. Fig. \ref{fig: 3D figure} shows all trajectories over the course of the entire simulation.

\begin{figure}[h]
    \centering
    \includegraphics[width=1\linewidth]{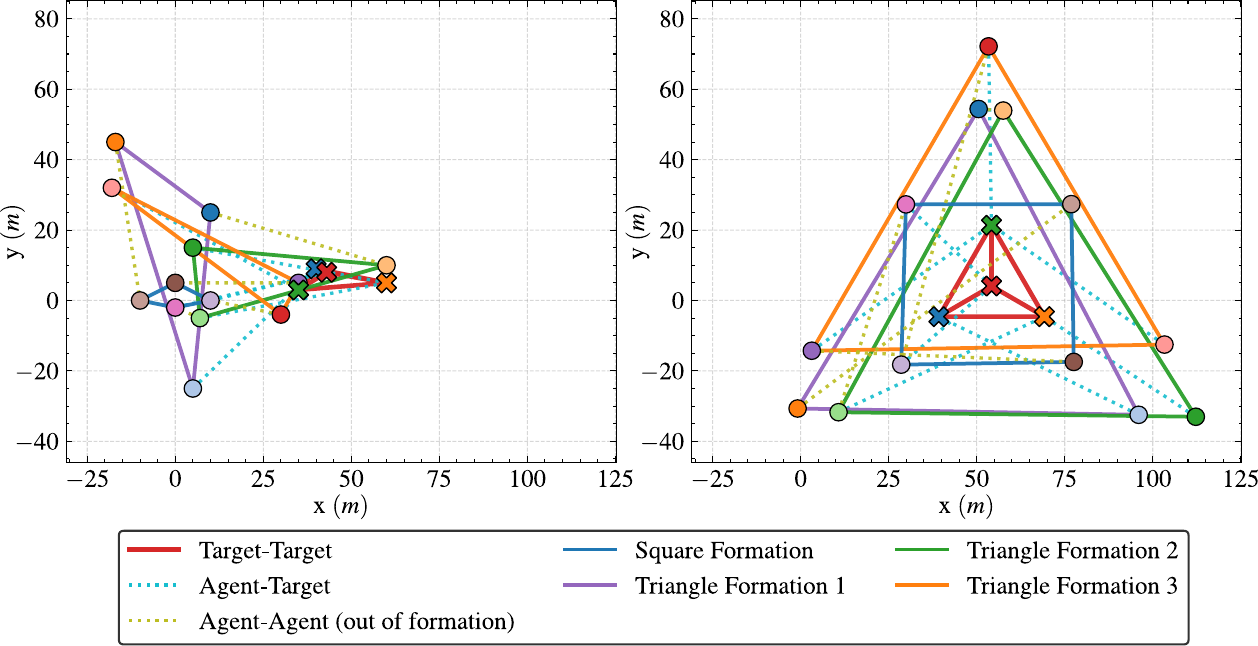}
    \caption{Comparison of the initial and final conditions of the simulation, with agent proportional gains set to $10$.}
    \label{fig: subplots}
\end{figure}

\begin{figure}[h]
    \centering
    \includegraphics[width=1\linewidth]{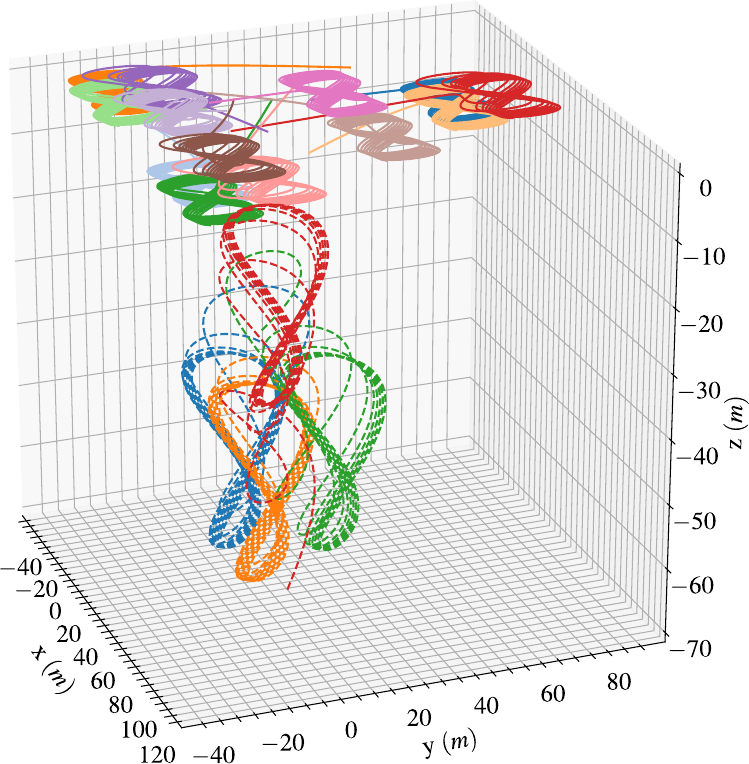}
    \caption{Visualization of all trajectories over the course of the entire simulation.}
    \label{fig: 3D figure}
\end{figure}

\section{Conclusion}
This work presented a novel framework for multi-agent, multi-target tracking capable of handling heterogeneous nonlinear agent dynamics and state spaces by leveraging the mathematical structure of cellular sheaves. Extending existing cooperative coordination frameworks to the non-cooperative tracking problem, this work formulated the task as a harmonic extension problem on a cellular sheaf. This formulation led to the development of a decentralized control law using the sheaf Laplacian, with a corresponding Lyapunov-based stability analysis guaranteeing tracking error convergence. Future research will focus on expanding this sheaf-theoretic approach to include nonlinear restriction maps, allowing the model to capture more complex agent interactions, and adapting the framework to systems defined over directed communication topologies.

\bibliographystyle{ieeetr}
\bibliography{master.bib}

\end{document}

%% file: preamble/preamble.tex
\input{preamble/packages}

\input{preamble/macros/optimization_macros}

\input{preamble/comments}

\input{preamble/environments}


%% file: preamble/packages.tex
\usepackage{graphicx} 
\usepackage{quiver}
\usepackage{verbatim}
\usepackage{todonotes}
\usepackage{xcolor}

\usepackage{amsmath, amsfonts, amssymb, amsthm, mathtools}
\usepackage{stmaryrd}
\let\labelindent\relax
\usepackage{enumitem}
\usepackage{thmtools}
\usepackage{fancyhdr}
\usepackage{hyperref}
\usepackage{tikz}
\usepackage{comment}
\usepackage[normalem]{ulem}
\usepackage{booktabs}
\usepackage{subcaption}  
\usepackage[font=footnotesize]{caption}  

%% file: preamble/macros/optimization_macros.tex
\newcommand{\maps}{\colon}
\newcommand{\R}{\mathbb{R}}

\newcommand{\define}[1]{\emph{#1}}

\newcommand{\FF}{\mathcal{F}}

\DeclareMathOperator{\face}{\trianglelefteqslant}
\DeclareMathOperator{\image}{im}

%% file: preamble/comments.tex
\todostyle{jpf}{color=NavyBlue!20, inline, author=James}
\todostyle{kev}{color=teal!20, inline, author=Kevin}
\todostyle{llm}{color=NavyBlue!60, inline, author=Luke}
\todostyle{gr}{color=violet!20, inline, author=George}
\todostyle{matt}{color=green!20, inline, author=MK}
\todostyle{matth}{color=NavyBlue!40, inline, author=MH}
\todostyle{tyler}{color=purple!60, inline, author=TH}
\todostyle{cristian}{color=red!60, inline, author=CN}
\todostyle{hans}{color=orange!20, inline, author=HR}
\todostyle{jbb}{color=pink!20, inline, author=JBB}
\todostyle{sam}{color=pink!20, inline, author=SC}
\todostyle{trevor}{color=pink!20, inline, author=TG}

%% file: preamble/environments.tex
\newtheorem{theorem}{Theorem}

\newtheorem{lemma}{Lemma}

\theoremstyle{definition}
\newtheorem{assumption}{Assumption}
\newtheorem{definition}{Definition}

\newtheorem{example}{Example}
\newtheorem{remark}{Remark}
\newtheorem{problem}{Problem}